\documentclass[draftcls,11pt,onecolumn,romanappendices]{IEEEtran}
\usepackage{times}
\usepackage{amsmath}  
\usepackage{amssymb} 
\usepackage{mathrsfs}
\usepackage{cite}
\usepackage{comment} 
\usepackage{upref}
\usepackage{amsfonts}
\usepackage{verbatim}
\usepackage[dvipsnames,usenames]{color}
\usepackage{enumerate}
\usepackage{graphicx}
\usepackage{subfigure}
\usepackage{latexsym}
\usepackage{bm}
\usepackage{multirow}
\usepackage{dsfont}
\usepackage{amsthm,xpatch}
\usepackage{mathtools}
\usepackage{bbm}
\usepackage{latexsym}
\usepackage{accents}
\usepackage{psfrag}
\usepackage{color}
\usepackage{times,color}
\usepackage[usenames,dvipsnames,svgnames,table]{xcolor}
\usepackage{pgf}
\usepackage{tikz}
\usetikzlibrary{arrows, automata, positioning, calc}
\usepackage{cancel} 

\newtheorem{lemma}{Lemma}
\newtheorem{theorem}{Theorem}
\newtheorem{corollary}{Corollary}
\newtheorem{remark}{Remark}

\usepackage{algorithm,algpseudocode}

\makeatletter
\newcommand{\removelatexerror}{\let\@latex@error\@gobble}
\makeatletter
\newcommand{\proofpart}[2]{%
	\par
	\addvspace{\medskipamount}%
	\noindent\emph{Part #1: #2}\par\nobreak
	\addvspace{\smallskipamount}%
	\@afterheading
}
\makeatother

\newcommand{\abs}[1]{\lvert{#1}\rvert}
\newcommand{\mbf}[1]{\mathbf{#1}}
\newcommand{\mc}[1]{\mathcal{#1}}
\newcommand{\bc}[1]{\boldsymbol{\mathcal{#1}}}
\newcommand{\mbb}[1]{\mathbb{#1}}

\usepackage[inline]{trackchanges}
\usepackage{setspace}
\addeditor{alex}

\theoremstyle{definition}

\begin{document}

\pagestyle{plain}

\title{\LARGE \textbf{Multi-Server Private Linear Transformation with Joint Privacy}}

\author{\Large Fatemeh Kazemi and Alex Sprintson \\{\normalsize Dept. of ECE, Texas A\&M University, USA \big(E-mail: \{fatemeh.kazemi, spalex\}@tamu.edu\big)}\thanks{This material is based upon work supported by the National Science Foundation under Grants No.~1718658 and 1642983.}}

\maketitle 

\thispagestyle{empty}

\begin{abstract} 
This paper focuses on the Private Linear Transformation (PLT) problem in the multi-server scenario. In this problem, there are $N$ servers, each of which stores an identical copy of a database consisting of $K$ independent messages, and there is a user who wishes to compute $L$ independent linear combinations of a subset of $D$ messages in the database while leaking no information to the servers about the identity of the entire set of these $D$ messages required for the computation. We focus on the setting in which the coefficient matrix of the desired $L$ linear combinations generates a Maximum Distance Separable (MDS) code. We characterize the capacity of the PLT problem, defined as the supremum of all achievable download rates, for all parameters $N, K, D \geq 1$ and $L=1$, i.e., when the user wishes to compute one linear combination of $D$ messages. Moreover, we establish an upper bound on the capacity of PLT problem for all parameters $N, K, D, L \geq 1$, and leveraging some known capacity results, we show the tightness of this bound in the following regimes: (i) the case when there is a single server (i.e., $N=1$), (ii) the case when $L=1$, and (iii) the case when $L=D$.

\end{abstract}


\section{Introduction}\label{sec:Intro}
 
\subsection{Motivation}

This work focuses on the  Private Linear Transformation (PLT) problem, recently introduced in~\mbox{\cite{HES2021Individual,HES2021Joint}}. In the PLT problem, there are $N$ servers, each of which stores an identical copy of a database consisting of $K$ independent messages. Also, there is a user who wishes to compute $L$ independent linear combinations of a subset of $D$ messages in the database, without revealing any information to the servers about the identities of the $D$ messages required for the computation, while downloading the minimum possible amount of information from the servers.

The PLT problem can be viewed as an interesting extension of the Private Information Retrieval (PIR) (see e.g.,~\cite{chor1995and,Sun2017,shah2014one,tajeddine2018,banawan2018capacity,BU17,BU2018,HKRS2019,LG:2018,chen2020capacity,TSC2018,siavoshani2021private,kadhe2019private,HKS2020,KKHS32019,HKS2018,KKHS12019,HKS2019,KKHS22019}) and Private Linear Computation (PLC) (see e.g.,~\cite{sun2018capacity,mirmohseni2018private,HS2019,HS2020}) problems, which have been extensively studied  in the literature. To be more specific, for $L=D$, the PLT problem reduces to the multi-message PIR problem in which the goal is to privately retrieve a subset of $D$  messages in the database. Moreover, for $L=1$, the PLT problem reduces to the PLC problem in which the goal is to privately compute one linear combination of a $D$-subset of messages. The PLT problem can be motivated by several practical scenarios such as linear transformation technique applied for dimensionality reduction in Machine Learning (ML) applications~(see \cite{HES2021Joint}). 

\subsection{Previous and Related Work}

In the classical PIR problem, a user wants to download a message from a database replicated over $N$ non-colluding servers, without leaking any information about the identity of the desired message to any individual server. The capacity of the information-theoretic PIR was derived in~\cite{Sun2017}. Then, the PIR problem has been extended in various directions, such as coded PIR (see e.g.,~\cite{shah2014one,tajeddine2018,banawan2018capacity}), multi-message PIR (see e.g.,~\cite{BU17,BU2018,HKRS2019,LG:2018}), and PIR with side information (see e.g.,~\cite{kadhe2019private,HKS2020,KKHS32019,HKS2018,KKHS12019,HKS2019,KKHS22019}). 

The problem of Private Computation (PC), initially introduced in~\cite{sun2018capacity}, is an interesting generalization of the PIR problem, in which the user wishes to compute one arbitrary linear combination of the messages in the database, while revealing no information about the identities and the coefficients of these messages to any server. Several variants of the PC problem were also studied in~\cite{obead2018achievable,obead2018,tahmasebi2019private,ASK2020,obead2020private,ASK2019,mirmohseni2018private,HS2019,HS2020}. In~\cite{HS2020}, a variation of the PC problem was considered in which it is only required to protect the identities of the messages in the demanded linear combination, while the coefficients used to construct the linear combination do not need to be hidden from the server.

The most related to this work is the PLT problem, recently introduced in~\cite{HES2021Individual, HES2021Joint}, which is also closely related to the PIR and PLC problems. Indeed, a naive protocol for the PLT problem is to privately retrieve all the $D$ messages required for the computation using a multi-message PIR scheme, and then compute the required linear combinations. Another simple approach for the PLT problem is to compute each required linear combination separately using a PLC protocol. 

Although there is a significant body of  literature on the PIR and PLC problems, there are only a few studies on the PLT problem. In particular, the PLT problem was studied in the single-server setting by considering the following two privacy requirements: (i) the individual privacy, where the identity of each individual message in the support set of the demanded linear combinations needs to be kept private~\cite{HES2021Individual}; and (ii) the joint privacy, in which the identity of the entire set of messages in the support set of the demanded linear combinations must be kept private~\cite{HES2021Joint}. All variants of the PIR and PC problems, can also be considered for the PLT problem which opens several interesting directions for future work.

In~\cite{HES2021Joint}, Heidarzadeh et al. recently proved that the capacity of the PLT with a single server and joint privacy is $L/(K-D+L)$. However, the capacity of the PLT in the multi-server scenario was left as an open problem in~\cite{HES2021Joint}. Remarkably, neither a general achievability scheme nor a converse was known in this case. This work is motivated by this open problem.

\subsection{Main Contributions} 

In this paper, we consider the multi-server setting of the PLT problem with an arbitrary number of servers ${N\geq 1}$. We focus on the setting in which the coefficient matrix of the required linear combinations generates a Maximum Distance Separable (MDS) code. This setting can be motivated by several practical scenarios. For instance, the user may have chosen the the coefficient matrix randomly over the field of real numbers or a finite field of large size~\cite{HES2021Joint}. The first contribution of this work is to show that the capacity of PLT problem for the case of $L=1$, i.e., when the user wishes to compute one linear combination of $D$ messages, is equal to $\Phi(1/N,K-D+1)$, where $\Phi(A,B)={(1+A+A^2+\cdots+A^{B-1})}^{-1}$. This result establishes the capacity of the PLC problem for an arbitrary number of servers $N$, thus settling the open problem mentioned above for the case of $L=1$. Moreover, we establish an upper bound on the capacity of PLT problem for any arbitrary parameters $N, K, D, L \geq 1$, and based on some known capacity results, we show the tightness of the provided upper bound for some special cases of the problem: (i) the case where there is a single server (i.e., $N=1$), (ii) the case where $L=1$, and (iii) the case where $L=D$.

\section{Problem Formulation}\label{sec:Problem-Formulation}

\subsection{Basic Notation}

Throughout this paper, we denote random variables by bold letters and their realizations by regular letters. The functions $\mbb{P}(\cdot)$, $\mbb{P}(\cdot|\cdot)$, $\mbb{H}(\cdot)$,  $\mbb{H}(\cdot|\cdot)$, and $I(\cdot;\cdot |\cdot)$ denote probability, conditional probability, entropy, conditional entropy, and conditional mutual information, respectively. Let ${\mbb{Z}_{\geq 0}}$ and $\mbb{N}$ denote the set of non-negative integers and the set of positive integers, respectively. For any $i\in \mbb{N}$, let $[i]\triangleq \{1,\dots,i\}$. Let $\mbb{F}_q$ be a finite field for some prime $q$, $\mathbb{F}_q^{\times} \triangleq \mathbb{F}_q\setminus \{0\}$ be the multiplicative group of $\mathbb{F}_q$, and $\mbb{F}^S_{q}$ be the $S$-dimensional vector space over $\mbb{F}_q$ for some integer $S\geq 1$. Let $B\triangleq S\log_2 q$. Let $K, D, L \geq 1$ be integers such that $L\leq D\leq K$. Let $\mc{K} \triangleq [K]$. Let $\mbb{W}$ denote the set of all $D$-subsets (i.e., subsets of size $D$) $\mc{W}$ of $\mc{K}$, and $\mbb{V}$ denote the set of all MDS matrices $V$ of dimension $L \times D$ with entries in $\mbb{F}_q$ (i.e., every $L \times L$ submatrix of $V$ is full-rank). We denote the cardinality of a set $\mc{S}$ by $\abs{\mc{S}}$. For a positive real number $A$ and a positive integer number $B$, let  $\Phi (A,B)={(1+A+A^2+\cdots+A^{B-1})}^{-1}$.

\subsection{Setup and Assumptions}

Consider $N$ non-colluding servers, each stores an identical copy of a database consisting of $K$ messages, $X_{\mc{K}}=\{X_1,\dots,X_K\}$, where each message $X_i$ is a row vector of length $S$. Let $X\triangleq[X_1^\top,\cdots,X_K^\top]^\top$ be a matrix of dimension $K \times S$. For some $\mc{R} \triangleq \{i_1,\dots,i_r\}\subset\mc{K}$, let $X_\mc{R}$ be the submatrix of $X$ of size $\abs{\mc{R}}\times S$, restricted to its rows indexed by the set $\mc{R}$, i.e., $X_\mc{R}=[X_{i_1}^\top,\cdots,X_{i_r}^\top]^\top$.

Suppose that there is a user who wishes to compute $L$ linear combinations of $D$ messages $\{X_i: i \in \mc{W}\}$, as $V_1 X_\mc{W},\cdots,V_L X_\mc{W}$, where $\mc{W}\in \mbb{W}$ is the index set of the $D$ messages required for the computation, and $V_\ell$, $\ell \in [L]$, denoting the coefficient vector of the $\ell$th desired linear combination, is the $\ell$th row of an $L\times D$ MDS matrix $V$ with entries in $\mbb{F}_q$, i.e., $V=[V_1^\top, \cdots, V_L^\top]^\top$, $V\in \mbb{V}$. In other words, the user wants to compute the $L \times S$ matrix $Z^{[\mc{W},V]}\triangleq VX_\mc{W}$ whose rows are the $L$ required linear combinations. We refer to $Z^{[\mc{W},V]}$ as the \emph{demand}, $\mc{W}$ as the \emph{demand's index set}, $V$ as the \emph{demand's coefficient matrix}, $L$ as the \emph{demand's dimension}, and  $D$ as the \emph{demand's support size}.


We assume that $\mbf{X}_1,\cdots,\mbf{X}_K$ are independently and uniformly distributed over $\mbb{F}^{S}_{q}$, i.e., ${H(\mbf{X}_i) = B}$ for $i\in \mc{K}$. Thus, $H(\mbf{X}) = KB$, $H(\mbf{X}_\mc{R}) = \abs{\mc{R}}B$ for every $\mc{R} \subset\mc{K}$, and $H(\mbf{Z^{[\bc{W},V]}}) = LB$. We also assume that $\bc{W}$, $\mbf{V}$, and $\mbf{X}$ are independent random variables such that $\bc{W}$ and $\mbf{V}$ are uniformly distributed over $\mbb{W}$ and $\mbb{V}$, respectively. Moreover, we assume that the servers initially know the distributions of $\bc{W}$ and $\mbf{V}$, whereas the servers have no information about the realizations $\mc{W}$ and $V$ in advance.

\subsection{Privacy and Recoverability Conditions}

To retrieve the demand $Z^{[\mc{W},V]}$ for any given $\mc{W}$ and $V$, the user generates $N$ queries $\{Q_n^{[\mc{W},V]}\}_{n\in [N]}$, and sends the query $Q_n^{[\mc{W},V]}$ to the $n$-th server. Note that server $n$ just receives $Q_n^{[\mc{W},V]}$ without having any access to other queries (non-colluding servers assumption). Each query $Q_n^{[\mc{W},V]}$ is a (potentially stochastic) function of $\mc{W}$ and $V$. For clarity, we denote ${Q}^{[\mc{W},V]}\triangleq\{Q_n^{[\mc{W},V]}\}_{n\in [N]}$ and $\mbf{Q}^{[\bc{W},\mbf{V}]}\triangleq\{\mbf{Q}_n^{[\bc{W},\mbf{V}]}\}_{n\in [N]}$.

Once the $n$-th server receives the query $Q_n^{[\mc{W},V]}$, it responds back to the user with an answer $A_n^{[\mc{W},V]}$. The answer $A_n^{[\mc{W},V]}$ is a (deterministic) function of the query $Q_n^{[\mc{W},V]}$ and $X$, i.e.,
$H(\mbf{A}_n^{[\bc{W},\mbf{V}]}| \mbf{Q}_n^{[\bc{W},\mbf{V}]},\mbf{X}) = 0$. For clarity, we denote $A^{[\mc{W},V]}\triangleq\{A_n^{[\mc{W},V]}\}_{n\in [N]}$  and $\mbf{A}^{[\bc{W},\mbf{V}]}\triangleq\{\mbf{A}_n^{[\bc{W},\mbf{V}]}\}_{n\in [N]}$.

\emph{\textbf{Recoverability Condition}}: The answers $A^{[\mc{W},V]}$ from all the servers along with the queries ${Q}^{[\mc{W},V]}$, and the realizations $\mc{W},V$ must enable the user to retrieve the demand ${Z^{[{W},V]}}$. This condition is referred to as the \emph{recoverability condition}, as formally stated in the following
\[H(\mbf{Z^{[\bc{W},V]}}| \mbf{A}^{[\bc{W},\mbf{V}]},\mbf{Q}^{[\bc{W},\mbf{V}]}, \bc{W},\mbf{V})=0,\]

\emph{\textbf{Privacy Condition}}: The queries ${Q}^{[\mc{W},V]}$ should be designed such that the servers infer no information about the user's demand index set $\mc{W}$. This condition is referred to as the \emph{joint privacy condition}, formally stated as follows
\[I(\bc{W}; \mbf{Q}_n^{[\bc{W},\mbf{V}]},\mbf{A}_n^{[\bc{W},\mbf{V}]},\mbf{X}_{\mc{K}})=0\quad \forall n \in [N].\] 
Equivalently, from the perspective of each server, every $D$-subset of indices $\mc{K}$ must be equally likely to be the demand’s index set, i.e., for any given $\mc{\tilde{W}} \in \mbb{W}$, it must hold that
\[\mathbb{P}(\bc{W}= \mc{\tilde{W}}| \mbf{Q}_n^{[\bc{W},\mbf{V}]}= {Q}_{n}^{[\mc{W},{V}]}) = \mathbb{P}(\bc{W}= \mc{\tilde{W}})\quad \forall n \in [N].\]

\subsection{Problem Statement}
The problem is to design a protocol for generating queries $\{Q_n^{[\mc{W},V]}\}_{n\in [N]}$ and their corresponding answers $\{A_n^{[\mc{W},V]}\}_{n\in [N]}$ (for any given $\mc{W}$ and ${V}$) such that both the privacy and recoverability conditions are satisfied. We refer to this problem as \emph{Private Linear Transformation (PLT)}. A protocol for generating queries/answers for PLT is referred to as a \emph{PLT protocol}. 

The \emph{rate} of a PLT protocol is defined as the ratio of the entropy of the demand , i.e., $H(\mbf{Z^{[\bc{W},V]}}) = LB$, to the total entropy of answers from the servers, i.e., $\Sigma_{n=1}^{N}H(\mbf{A}_{n}^{[\bc{W},\mbf{V}]})$. The \emph{capacity} of the PLT problem, denoted by ${C^{PLT}(N,K,L,D)}$, is defined as the supremum of rates over all PLT protocols, i.e.,
\[C^{PLT}(N,K,L,D)\triangleq \sup \frac{LB}{\Sigma_{n=1}^{N}H(\mbf{A}_{n}^{[\bc{W},\mbf{V}]})}\]
In this work, our goal is to characterize (or derive non-trivial bounds on) the capacity of the PLT problem, i.e., ${C^{PLT}(N,K,L,D)}$, and to design a PLT protocol that is capacity-achieving.

\section{Main Results}\label{sec:Main-Results}

In this section, we present our main results. Theorem~\ref{thm:PLT} establishes an upper bound on the capacity of the PLT problem for all parameters $N, K, L, D \geq 1$. Leveraging some known capacity results, we show that the presented upper bound is tight in the following regimes: (i) the case where there is a single server (i.e., $N=1$), (ii) the case where $L=1$, and (iii) the case where $L=D$. Theorem~\ref{thm:PLC} characterizes the capacity of the PLT problem for all parameters $N, K, D \geq 1$ and $L=1$, i.e., the case where the user wishes to privately compute \emph{one} linear combination of $D$ messages in the database. The proofs of theorems~\ref{thm:PLT} and~\ref{thm:PLC} are given in sections~\ref{sec:PLT} and~\ref{sec:PLC}, respectively.

\begin{theorem}\label{thm:PLT}
The capacity of the PLT problem with $N$ non-colluding and replicated servers, $K$ messages, demand's support size $D$, and demand’s dimension $L$,   \newline 
(i) if $\frac{K-D}{L}\leq 1$, is upper bounded by 
\[{C^{PLT}(N,K,L,D)} \leq {\left(1+\frac{K-D}{LN}\right)}^{-1},\] 
(ii) and if $\frac{K-D}{L}\geq 1$, is upper bounded by
\[{C^{PLT}(N,K,L,D)} \leq {\left(\frac{1-{\left({\frac{1}{N}}\right)}^{\lfloor{\theta\rfloor}}}{1-\frac{1}{N}}+\frac{\left(\theta-{\lfloor{\theta\rfloor}}\right)}{{N}^{\lfloor{\theta\rfloor}}}\right)}^{-1}.\]where $\theta\triangleq{{\frac{K-D+L}{L}}}$.
\end{theorem}

The converse proof is provided in Section~\ref{subsec:PLT:Converse}, which is based on a reduction argument and leverages the capacity result for multi-message PIR with private side information problem, introduced in~\cite{siavoshani2021private}. 


\begin{corollary}\label{cor:PLT}
If $\frac{K-D}{L}\in {\mbb{Z}_{\geq 0}}$, the capacity upper bounds provided in Theorem~\ref{thm:PLT}, can be written as\[{C^{PLT}(N,K,L,D)} \leq \left(1+\frac{1}{N}+\dots+\frac{1}{N^{\frac{K-D}{L}}}\right)^{-1}=\Phi(\frac{1}{N},\frac{K-D+L}{L}).\]
\end{corollary}

\begin{remark}\label{rem:1}
\emph{The capacity upper bounds in Theorem~\ref{thm:PLT} are tight for the case when $N=1$ (i.e., when there is a single server), which is equal to $L/(K-D+L)$ as was shown in~\cite[Theorem~2]{HES2021Joint}. Moreover, in Theorem~\ref{thm:PLC}, we prove the tightness of this upper bound for the case of $L=1$.}
\end{remark}

\begin{remark}\label{rem:2}
\emph{Notably, for the case of $L=D$, where the user wishes to privately compute $D$ independent linear combinations of $D$-subset of messages in the database (which is equivalent to privately retrieving these $D$ messages), the capacity upper bound in Theorem~\ref{thm:PLT}, i.e., (i) ${{\left(1+{(K-D)}/{DN}\right)}^{-1}}$ if ${K/D \leq 2}$, and (ii) $\Phi(1/N,K/D)$ if $K/D \geq 2$ and $K/D \in \mbb{N}$, is tight as was shown in~\cite{BU2018}. Note that in this case, an optimal capacity-achieving multi-message PIR protocol proposed in~\cite[Theorems~1,~2]{BU2018} is an optimal protocol that achieves the capacity upper bound in Theorem~\ref{thm:PLT}.}
\end{remark}

\begin{theorem}\label{thm:PLC}
The capacity of the PLT problem with $N$ non-colluding and replicated servers, $K$ messages, demand's support size $D$, and demand’s dimension $L=1$,  is given by
\[{C^{PLT}(N,K,1,D)} = \left(1+\frac{1}{N}+\dots+\frac{1}{N^{K-D}}\right)^{-1}=\Phi\left(\frac{1}{N},K-D+1\right).\]
\end{theorem}

The converse proof follows directly from the result of Theorem~\ref{thm:PLT} for $L=1$. Also, an alternative proof of converse, similar to that of Theorem~\ref{thm:PLT}, is provided in Section~\ref{sec:PLC}. For the achievability proof, we design a PLT protocol that achieves the proposed upper bound provided by converse, and is inspired by both our recently proposed scheme of~\cite{HKS2019} for the single-server PIR with private coded side information problem, and the scheme proposed in~\cite{sun2018capacity} for the private computation problem.

\begin{remark}\label{rem:3}
\emph{The result of Theorem~\ref{thm:PLC} generalizes the previous finding reported in~\cite{HES2021Joint} for the PLT problem with a single server, without any prior side information, when joint privacy is required, and $L=1$. As was shown in~\cite{HES2021Joint}, the capacity of this setting is equal to $K-D+1$, which is consistent with the result of Theorem~\ref{thm:PLC} for $N=1$. Also, evidently it can be observed that for the case of $D=1$, the result of Theorem~\ref{thm:PLC} reduces to the known capacity result of~\cite{Sun2017} for the classical PIR problem where the user wants to privately download one message in the database, which is $\Phi\left(1/N,K\right)$}.
\end{remark}

\begin{remark}\label{rem:4}
\emph{It is worthwhile to compare the result of Theorem~\ref{thm:PLC} with the capacity result of~\cite{sun2018capacity} for the related PC problem where the user wishes to compute one arbitrary linear combination of $K$ independent messages in a database replicated at $N$ non-colluding servers, while hiding both the identities and the coefficients of the messages participating in the demand. As was shown in~\cite{sun2018capacity}, the capacity of this setting is equal to $\Phi\left(1/N,K\right)$. Unlike the privacy requirements in the private computation problem introduced in~\cite{sun2018capacity}, in the PLT problem, the goal is to hide only the identities of the $D$ messages participating in the user's demand and not necessarily the values of their coefficients, which based on the result of Theorem~\ref{thm:PLC}, it can be fulfilled more efficiently with much higher rate, i.e., $\Phi\left(1/N,K-D+1\right)$. This is interesting since this type of access privacy are motivated by many practical scenarios such as linear transformation technique used for dimensionality reduction in Machine Learning (ML) applications (see, e.g.~\cite{HES2021Joint,bingham2001random} and references therein). By comparing the capacity results of these two problems, one can readily conclude that the advantage of PLT protocols over the a repeated use of a PC protocol becomes more tangible when the demand's support size $D$ increases.}
\end{remark}

\begin{remark}\label{rem:5}
\emph{It is noteworthy that for\footnote{Note that for the case of $D=1$, the PLT problem reduces to the classical single-message PIR problem introduced in~\cite{Sun2017}.} $D \geq 2$, a trivial PLT protocol for $L=1$ would be privately retrieving the $D$ messages required for the linear computation using an optimal multi-message PIR scheme satisfying privacy of demand messages jointly, introduced in~\cite{BU2018}, and then computing the required linear combination. As was shown in~\cite[Theorems~1,~2]{BU2018}, the optimal rate that can be achieved leveraging this approach, is upper bounded by $D^{-1}\leq 1/2$. The result of Theorem~\ref{thm:PLC} indicates that the PLT problem in general can be addressed much more efficiently with the rate of $\Phi(1/N,K-D+1)\geq 1/2$.}
\end{remark}

\begin{remark}\label{rem:6}
\emph{Interestingly, in the PLT problem, a simple approach of computing each of the required linear combinations separately through applying an optimal PLT scheme introduced in Theorem~\ref{thm:PLC}, cannot achieve the capacity upper bound presented in Theorem~\ref{thm:PLT} for all parameters $N, K, L, D$.}
\end{remark}

\section{Proof of Theorem~\ref{thm:PLT}}\label{sec:PLT}
\subsection{Converse proof}\label{subsec:PLT:Converse}
The proof of converse follows from the capacity result for the problem of multi-message PIR with private side information, referred to as M-PIR-PSI, introduced in~\cite[Theorem~1]{siavoshani2021private}. In this problem, there is a database of $K$ independent messages whose copies are replicated across $N$ servers, and there is a user who has access to $M$ messages from the database as side information. The user wishes to retrieve $P$ messages from the database while leaking no information about the the identities of both the desired messages and the side information messages, to any individual server. As was shown in~\cite[Theorem~1]{siavoshani2021private}, the capacity of this setting, denoted by $C^{MPIR-PSI}(N,K,P,M)$,\newline
(i) if $\frac{K-M}{P}\leq 2$ is given by
\begin{equation}\label{upp:(i)}
{C^{MPIR-PSI}(N,K,P,M)} = {\left(1+\frac{K-M-P}{PN}\right)}^{-1},
\end{equation}
(ii) if $\frac{K-M}{P}\geq 2$ is upper bounded by
\begin{equation}\label{upp:(ii)}
{C^{MPIR-PSI}(N,K,P,M)} \leq {\left(\frac{1-{\left({\frac{1}{N}}\right)}^{\lfloor{\rho\rfloor}}}{1-\frac{1}{N}}+\frac{\left(\rho-{\lfloor{\rho\rfloor}}\right)}{{N}^{\lfloor{\rho\rfloor}}}\right)}^{-1},
\end{equation}
where $\rho\triangleq{{\frac{K-M}{P}}}$. In case (ii), as was shown~\cite[Corollary~1]{siavoshani2021private}, if $\tfrac{K-M}{P} \in \mbb{N}$, the capacity is given by
\begin{equation}\label{upp:(ii-1)}
C^{MPIR-PSI}(N,K,P,M)=\Phi(\frac{1}{N},\frac{K-M}{P}).
\end{equation}

In the following, we want to show that any PLT protocol designed for the problem with $N$ servers, $K$ messages, demand's support size $D$, and demand’s dimension $L$, can be used as a protocol that satisfies both the recoverability and the privacy conditions of the M-PIR-PSI problem with demand size ${P=L}$ and side information size ${M=D-L}$. Specifically, for a given instance of the M-PIR-PSI problem with the set of demand indices $\mc{P}$ of size $L$, (i.e., $P=L$), and the set of side information indices $\mc{S}$ of size $D-L$ , (i.e., ${M=D-L}$), the user can  construct a random $L\times D$ MDS matrix $V$ and forms the set $\mc{W}=\mc{P}\cup\mc{S}$. Then, for the given $\mc{W}$ and $V$, the user and the servers can apply a PLT protocol for generating queries $Q^{[\mc{W},V]}$ and their corresponding answers $A^{[\mc{W},V]}$, such that the user can privately compute $L$ MDS coded linear combinations of the $D$ messages indexed by the set $\mc{W}$ (i.e., union of demands and side information messages). The user can then retrieve the $L$ desired messages by subtracting off the contribution of the $D-L$ side information messages from the computed $L$ linear combinations. 


Now, we need to prove that the PLT-based protocol described above satisfies both the recoverability and the joint privacy conditions of the M-PIR-PSI problem. It should be noted that since the PLT protocol enables the user to compute $L$ MDS coded linear combinations of $D$ messages, based on the property of MDS codes\footnote{Every $L \times L$ submatrix of an $L \times D$ MDS matrix is invertible.}, one can readily verify that the user can always retrieve the $L$ desired messages by subtracting off the contribution of $D-L$ side information messages from the $L$ computed linear equations, and solving the resulting system of $L$ linear equations with $L$ unknowns. Thus, the recoverability condition is satisfied. 

It is easy to verify that by applying the PLT protocol, the identities of all the $D$ messages (i.e., the union of the demand messages and side information messages) participating in the $L$ linear combinations, will be jointly protected from each server as a result of the privacy guarantees of the PLT protocol. Indeed, from the perspective of each server, every $D$-subset of $K$ messages is equally likely to be the union of the demand messages and side information messages. Moreover, due to the property of MDS codes, within each $D$-subset of messages, every subset of size $L$ can be considered as the set of demand messages (i.e., the remaining $D-L$ as the set of side information messages) with equal probability. This ensures that the described PLT-based protocol satisfies the privacy condition in the M-PIR-PSI problem.

Thus, we conclude that any achievable rates in the PLT problem with $N$ servers, $K$ messages, demand's support size $D$, and demand’s dimension $L$, would be also achievable (using the PLT-based protocol) in the M-PIR-PSI problem with $N$ servers, $K$ messages, demand size ${P=L}$, and side information size ${M=D-L}$. Thus, the capacity of PLT problem with parameters $N, K, D, L$, i.e., $C^{PLT}(N,K,L,D)$, is upper bounded by the capacity of the M-PIR-PSI problem with parameters $N, K , P=L, M=D-L$, i.e., ${C^{MPIR-PSI}(N,K,L,D-L)}$. Thus, substituting $P$ with $L$, and $M$ with $D-L$ in equations~\ref{upp:(i)},~\ref{upp:(ii)} completes the proof. Also, in case (ii), if ${\tfrac{K-M}{P}=\tfrac{K-D+L}{L} \in \mbb{N}}$ or equivalently  $\tfrac{K-D}{L} \in \mbb{Z}_{\geq 0}$, we have
\[C^{PLT}(N,K,L,D) \leq C^{MPIR-PSI}(N,K,L,D-L)=\Phi(\frac{1}{N},\frac{K-D+L}{L}).\]

\section{Proof of Theorem~\ref{thm:PLC}}\label{sec:PLC}

Here, we prove the converse by showing that the capacity for the case of $L=1$, i.e., ${C^{PLT}(N,K,1,D)}$, is upper bounded by the capacity of PIR with private side information problem, referred to as PIR-PSI, in which a database of $K$ independent messages is replicated across $N$ servers, and the user has access to $M$ messages from the database as side information. The user wants to retrieve one message from the database while hiding jointly the identities of the desired message and the side information messages, from any individual server. This problem was introduced by Chen et al.~\cite{chen2020capacity}. As was shown in~\cite[Theorem~1]{chen2020capacity}, the capacity of PIR-PSI problem, denoted by ${C^{PIR-PSI}(N,K,M)}$, is equal to $\Phi(\frac{1}{N},{K-M})$. 

Any PLT protocol designed for the problem with $N$ servers, $K$ messages, demand's support size $D$, and demand’s dimension $L=1$, enables the user to compute one linear combination of a subset of $D$ messages while hiding the identities of these messages from any server. So, based on a similar reasoning used in the converse proof of Theorem~\ref{thm:PLT}, one can easily confirm that such PLT protocol would also be a protocol satisfying the recoverability and the privacy conditions in the PIR-PSI problem with side information size ${M=D-1}$. Thus, any achievable rate in the PLT problem with $N$ servers, $K$ messages, demand's support size $D$, and demand’s dimension $L=1$, can be also achieved for the PIR-PSI problem with $N$ servers, $K$ messages, and side information size ${M=D-1}$. Thus, we have
\[C^{PLT}(N,K,1,D) \leq C^{PIR-PSI}(N,K,D-1)=\Phi(\frac{1}{N},{K-D+1}).\]
\subsection{Achievability proof}\label{subsec:PLC:Achievability}
In this section, we complete the proof of Theorem~\ref{thm:PLC} by designing a PLT protocol for the setting with $N$ servers, $K$ messages, demand's support size $D$, and demand’s dimension ${L=1}$, such that it achieves the upper bound provided by converse on the rate of any such PLT protocols, i.e., ${\Phi(1/N,K-D+1)}$. The proposed protocol, referred to as the \emph{Modified GRS Code}, leverages ideas from a modified version of the Specialized GRS Code Protocol proposed for the problem of single-server PIR with private coded side information in~\cite{HKS2019}, as well as the PC scheme proposed for the PC problem in~\cite{sun2018capacity}. 

\textbf{Modified GRS Code protocol:} Assume $q\geq K$, and let each message consists of $S= N^{\binom{K}{D}}$ symbols from $\mathbb{F}_q$. Suppose the user wishes to privately compute one linear combination of $D$ messages indexed by a set $\mc{W}$, as $V_1 X_\mc{W}=\sum_{i\in \mc{W}} v_i X_i$ where $V_1$ is a row vector of length $D$. This protocol consists of four steps as follows:

\textit{\textbf{Step 1:}} By using the Modified Specialized GRS Code protocol proposed in~\cite{HKS2019}, the user first constructs a polynomial ${p(x) = \sum_{i=0}^{K-D} p_i x^i \triangleq \prod_{i\not\in \mc{W}} (x-\omega_i)}$ where $\omega_1,\dots,\omega_K$ are $K$ distinct arbitrarily chosen elements from $\mathbb{F}_q$. The user then constructs $r \triangleq K-D+1$ vectors ${{Q}_1,\dots,{Q}_{r}}$, each of length $K$, such that ${{Q}_{i} = [\alpha_1\omega_1^{i-1},\dots,\alpha_K\omega_K^{i-1}]}$, $i\in [r]$, where ${\alpha_j=\frac{v_j}{p(\omega_j)}}$ for any ${j\in \mc{W}}$, and $\alpha_j$ is chosen randomly from $\mathbb{F}^{\times}_q$ for any $j\not\in \mc{W}$.

\textit{\textbf{Step 2:}} Let $\hat{X}_i\triangleq \sum_{j=1}^{K} \alpha_j\omega_j^{i-1} X_{j}$ for $i\in [r]$. We refer to $\hat{X}_i$ as a \emph{super-message}. Note that the vector ${Q}_i$, constructed in Step 1, is the vector of coefficients of the messages $\{X_i\}_{i\in \mc{K}}$ in the super-message $\hat{X}_i$. Let $F\triangleq\binom{K}{D}$, and let $W_1,W_2,\dots,W_F$ be the collection of all $D$-subsets of $\mc{K}$ in a lexicographical order. The structure of the Specialized GRS Code protocol~\cite{HKS2019} ensures that for each $W_f$, $f \in [F]$, there exist exactly $q-1$ linear combinations $Y^1_f,Y^2_f,\dots,Y^{q-1}_f$ of the messages $\{X_i\}_{i \in W_f}$ with (non-zero) coefficients from $\mathbb{F}^{\times}_q$, such that for every $k\in [q-1]$, $Y^k_f$ can be written as a linear combination of the super-messages $\hat{X}_1,\dots,\hat{X}_r$. Let ${\beta}^{k}_f\triangleq [\beta^{k}_{f,1},\dots,\beta^{k}_{f,r}]$ be a vector of length $r$ such that $Y^{k}_f = \sum_{i=1}^{r} \beta^{k}_{f,i} \hat{X}_i$. It should be noted that, for each $f \in [F]$, $Y^1_f,Y^2_f,\dots,Y^{q-1}_f$ are the same up to a scalar multiple, i.e., for each $k\in [q-1]$, $Y^{k}_f = \delta_k Y^{1}_f$, or equivalently, ${\beta}^{k}_f = \delta_k {\beta}^{1}_f$, for some distinct $\delta_k\in \mathbb{F}^{\times}_q$. The user then constructs $F$ vectors ${\beta}_1,\dots,{\beta}_F$, each of length $r$, such that ${\beta}_f = {\beta}^{k_f}_f$ for $f \in [F]$, is chosen arbitrarily from the set of vectors $\{{\beta}^{k}_f\}_{k\in [q-1]}$. Let $Y_f\triangleq Y^{k_f}_f$ for $f\in [F]$. Each $Y_f$ is referred to as a (linear) \emph{function}. Note that ${\beta}_f$ is the vector of coefficients of the super-messages $\{\hat{X}_i\}_{i\in [r]}$ in the function $Y_f$.

\textit{\textbf{Step 3:}} The user then sends to all servers the vectors ${Q}_1,\dots,{Q}_{r}$, associated with the super-messages $\hat{X}_1,\dots,\hat{X}_r$, and the vectors ${\beta}_1,\dots,{\beta}_{F}$, associated with the functions $Y_1,\dots,Y_F$.

\textit{\textbf{Step 4:}} Then, the user and the servers leverage the PC scheme of~\cite{sun2018capacity} with $r$ (independent) messages and $F$ (linear) functions of these messages such that the user can privately retrieve one of these functions. Indeed, the $r=K-D+1$ super-messages $\{\hat{X}_i\}_{i\in [r]}$ and the $F$ functions $\{Y_f\}_{f\in [F]}$, respectively, play the role of the original messages and the functions in the PC scheme, and the user is interested in retrieving the function $Y_{f^{*}}$ privately, where $Y_{f^{*}}$ is a linear combination with non-zero coefficients of the messages $\{X_i\}_{i \in \mc{W}}$. Note that by construction, there exists only one function $Y_{f^{*}}$ among $Y_1,\dots,Y_F$ such that $Y_{f^{*}}$ is a linear combination (with only non-zero coefficients) of the messages $\{X_i\}_{i \in \mc{W}}$, and the user's demand is an scalar multiple of $Y_{f^{*}}$. More specifically, each server first constructs the super-messages $\{\hat{X}_i\}_{i\in [r]}$ by using the coefficient vectors $\{{Q}_i\}_{i\in [r]}$ as described in Step~2, and then constructs the functions $\{Y_f\}_{f\in [F]}$ by utilizing the super-messages $\{\hat{X}_i\}_{i\in [r]}$ and the coefficient vectors $\{{\beta}_f\}_{f\in [F]}$ as explained in Step~2. Note that each function $Y_f$ for $f\in [F]$ consists of $S=N^F$ symbols (from $\mbb{F}_q$) where $N$ is the number of servers. Then, each server sends to the user $S({{1}/{N}+{1}/{N^2}+\dots+{1}/{N^{K-D+1}}})$ carefully designed linear combinations of all symbols associated with all functions $\{Y_f\}_{f\in [F]}$. The details of the design of the user's query to each server and each server's transmitted linear combinations (which also depend on the query of the user) can be found in~\cite[Section~4]{sun2018capacity}.

\textbf{\textbf{Example 1.}} (Modified GRS Code protocol) Assume that $K=4$ independent messages from $\mbb{F}^{16}_{5}$ are replicated over $N = 2$ servers, and the user wishes to compute one linear combination of $D=3$ messages as $2X_1+X_2+X_3$, i.e., $\mc{W} = \{1,2,3\}$ and $V_1=[2,1,1]$ (i.e., $v_1=2$, $v_2=1$, and $v_3=1$). Note that each message consists of ${S= N^{\binom{K}{D}}= 16}$ symbols from $\mathbb{F}_{5}$.

First, the user chooses ${K=4}$ distinct elements $\omega_1,\dots,\omega_4$ from $\mathbb{F}_5$. Suppose that the user picks $\omega_1=0$, $\omega_2=1$, $\omega_3=2$, $\omega_4=3$, and then constructs the polynomial ${{p(x) = \prod_{i\not\in \mc{W}} (x-\omega_i)}=x-\omega_4=x-3}$. Then, the user computes $\alpha_j$ for $j\in \mc{W}$, as follows; $\alpha_1=\frac{v_1}{p(\omega_1)}=1$, $\alpha_2=\frac{v_2}{p(\omega_2)}=2$ and $\alpha_3=\frac{v_3}{p(\omega_3)}=4$, and chooses $\alpha_j$ for $j\not\in \mc{W}$, i.e., $\alpha_4$, randomly from $\mathbb{F}^{\times}_5$. Assume that the user chooses $\alpha_4=2$. 

Then, the user constructs ${r=K-D+1=2}$ vectors ${Q}_1$ and ${Q}_2$, each of length ${K=4}$, such that ${{Q}_i=[\alpha_1\omega_1^{i-1},\dots,\alpha_K\omega_K^{i-1}]}$ for $i\in \{1,2\}$, i.e., the user constructs ${{Q}_1=[1,2,4,2]}$ and ${{Q}_2=[0,2,3,1]}$. Note that for the set $W_1=\{1,2,3\}$, there exist exactly ${q-1=4}$ vectors ${{\beta}^{k}_1 = [2k,k]}$ for ${k\in [4]}$ such that ${2k{Q}_1+k{Q}_2=k[2,1,1,0]}$. 

Then, the user arbitrarily chooses the vector ${{\beta}_1}$ from the set of vectors $\{{\beta}^{k}_1 = [2k,k]\}_{k \in [4]}$. Suppose that the user chooses ${{\beta}_1={\beta}^{2}_1=[4,2]}$. Similarly, the user picks the vectors ${{\beta}_2=[3,1]}$, ${{\beta}_3=[1,4]}$ and ${{\beta}_4=[0,3]}$. Then, the user sends to all servers the vectors ${Q}_1$ and ${Q}_{2}$ (associated with the super-messages $\hat{X}_1$ and $\hat{X}_2$), and the vectors ${{\beta}_1,\dots,{\beta}_{4}}$ (associated with the functions ${Y_1,\dots,Y_4}$). Using the coefficient vectors ${Q}_1$ and ${Q}_{2}$, each server first constructs the two super-messages ${\hat{X}_1=X_1+2X_2+4X_3+2X_4}$ and ${\hat{X}_2=2X_2+3X_3+X_4}$, and then constructs the functions ${Y_1,\dots,Y_4}$ using the super-messages $\hat{X}_1$ and $\hat{X}_2$ and the coefficient vectors ${{\beta}_1,\dots,{\beta}_{4}}$ as follows: 

\[
\begin{array}{lcl}
Y_1=4\hat{X}_1+2\hat{X}_2=4X_1+2X_2+2X_3\\
Y_2=3\hat{X}_1+\hat{X}_2=3X_1+3X_2+2X_4\\
Y_3=\hat{X}_1+4\hat{X}_2=X_1+X_3+X_4\\
Y_4=3\hat{X}_2=X_2+4X_3+3X_4\\   
\end{array}
\]

Finally, the user and the servers apply the PC scheme of~\cite{sun2018capacity} for two super-messages $\hat{X}_1$, $\hat{X}_2$ in order for the user to privately retrieve the function ${Y_1}$. It should be noted that among the functions $Y_1,\dots,Y_4$, only $Y_1$ is a linear combination of the messages $\{X_i\}_{i\in \mc{W}} = \{X_1,X_2,X_3\}$, and the user's demand, i.e., $2X_1+X_2+X_3$ is equal to $3Y_1$. The details of the PC scheme for this example are as follows. Let $\pi: [16]\rightarrow [16]$ be a randomly chosen permutation. Let $u_f(i)\triangleq\sigma_i Y_f(\pi(i))$ for $f \in [4]$ and ${i \in [16]}$, where $Y_f(\pi(i))$ is the $\pi(i)$-th $\mathbb{F}_5$-symbol of $Y_f$, and $\sigma_i$ is a randomly chosen element from $\{-1,+1\}$. For simplifying the notation, let $(a_i,b_i,c_i,d_i)=(u_1(i),u_2(i),u_3(i),u_4(i))$ for all ${i \in [16]}$. The user then queries $15$ carefully designed linear combinations of the symbols $\{\{a_i\}_{i\in [16]},\{b_i\}_{i\in [16]},\{c_i\}_{i\in [16]},\{d_i\}_{i\in [16]}\}$, as given in Table~\ref{table1} \cite{sun2018capacity}, from each of the servers (S1 and S2). 

As shown in~\cite{sun2018capacity}, among the $15$ symbols queried from S1 (or S2), $3$ symbols are redundant (based on the information obtained from S2 (or S1)). For example, consider the $15$ symbols queried from S1. (Similar observations can be made regarding the queries from S2.) Among the $4$ symbols $\{a_1,b_1,c_1,d_1\}$, any $2$ symbols suffice to recover the other $2$ symbols. For example, $c_1$ and $d_1$ can be obtained from $a_1$ and $b_1$. (Note that $Y_3$ and $Y_4$ can be written as a linear combination of $Y_1$ and $Y_2$.) Thus, the server S1 needs to send two arbitrary symbols from $\{a_1,b_1,c_1,d_1\}$. In addition, given any $2$ symbols from $\{a_2,b_2,c_2,d_2\}$, any $5$ symbols among the $6$ symbols $\{{a_3-b_2},{a_4-c_2},{a_5-d_2},{b_4-c_3},{b_5-d_3},{c_5-d_4}\}$ queried from S1 would suffice to recover the remaining symbol. For example, ${c_5-d_4}$ can be obtained from the symbols $\{a_3-b_2,{a_4-c_2},{a_5-d_2},{b_4-c_3},{b_5-d_3},b_2,d_2\}$ (for details, see~\cite[Section~5.1]{sun2018capacity}). 
Thus, each of the servers S1 and S2 needs to send to the user only $12$ symbols. In particular, S1 transmits $2$ arbitrary symbols from $\{a_1,b_1,c_1,d_1\}$, $5$ arbitrary symbols from $\{{a_3-b_2},{a_4-c_2},{a_5-d_2},{b_4-c_3},{b_5-d_3},{c_5-d_4}\}$, and the $4$ symbols $\{a_9-b_7+c_6,a_{10}-b_8+d_6,a_{11}-c_8+d_7,b_{11}-c_{10}+d_9\}$, and the symbol $\{a_{15}-b_{14}+c_{13}-d_{12}\}$; and S2 transmits $2$ arbitrary symbols from $\{a_2,b_2,c_2,d_2\}$, $5$ arbitrary symbols from $\{a_6-b_1,{a_7-c_1},{a_8-d_1},{b_7-c_6},{b_8-d_6},{c_8-d_7}\}$, and the $4$ symbols $\{a_{12}-b_4+c_3,a_{13}-b_5+d_3,a_{14}-c_5+d_4,b_{14}-c_{13}+d_{12}\}$, and the symbol $\{a_{16}-b_{11}+c_{10}-d_{9}\}$. \\

\begin{table}[t!]
    \caption{The queries of the PC protocol for $N=2$, $2$ super-messages, $F=4$, when the user demands ${Y_1}$\cite{sun2018capacity}.}
    \label{table1}
    \centering
    \scalebox{1.25}{
\begin{tabular}{ |c|c| } 
 \hline 
 S1 & S2 \\ 
 \hline
 $a_{1},b_{1},c_{1},d_{1}$ & $a_{2}, b_{2},c_{2},d_{2}$ \\ 
 \hline
 $a_{3}-b_{2}$ & $a_{6}-b_{1}$\\ 
 $a_{4}-c_{2}$ & $a_{7}-c_{1}$\\ 
 $a_{5}-d_{2}$ & $a_{8}-d_{1}$ \\ 
  $b_{4}-c_{3}$ & $b_{7}-c_{6}$\\ 
  $b_{5}-d_{3}$ & $b_{8}-d_{6}$\\ 
  $c_{5}-d_{4}$ & $c_{8}-d_{7}$\\ 
 \hline
  $a_{9}-b_{7}+c_{6}$ & $a_{12}-b_{4}+c_{3}$\\ 
  $a_{10}-b_{8}+d_{6}$ & $a_{13}-b_{5}+d_{3}$\\ 
  $a_{11}-c_{8}+d_{7}$ & $a_{14}-c_{5}+d_{4}$\\ 
  $b_{11}-c_{10}+d_{9}$ & $b_{14}-c_{13}+d_{12}$\\ 
  \hline
  $a_{15}-b_{14}+c_{13}-d_{12}$ & $a_{16}-b_{11}+c_{10}-d_{9}$\\ 
  \hline
\end{tabular}}
\end{table}

From the answers sent by the servers, the user obtains all $16$ symbols $a_1,\dots,a_{16}$, and accordingly, all $16$ symbols of $Y_1$. (Note that $a_i = u_1(i)=\sigma_i Y_1(\pi(i))$ for $i\in [16]$.) Then, the user can compute the desired linear combination, i.e., $2X_1+X_2+X_3$ by computing $3Y_1$. In order to retrieve $Y_1$ which consists of $16$ symbols (over $\mathbb{F}_5$), according to the proposed protocol, the user downloads $24$ symbols (over $\mathbb{F}_5$) from both servers. Thus, the rate of the proposed protocol is $16/24=2/3$. 

It should be noted that for every subset of size $3$ of the messages $\{X_i\}_{i\in [4]}$, in the proposed protocol, there exists one (and only one) linear combination (with non-zero coefficients) of these messages, namely $Y_{f^{*}}$ for some ${f^{*} \in [4]}$. Moreover, as a result of the privacy guarantees of the PC scheme, no server can infer any information about the index ($f^{*}$) of the function $Y_{f^{*}}$ being requested by the user. Thus, the proposed scheme satisfies the required joint privacy condition of the PLT problem.  

\begin{lemma}\label{lemma1}
The Modified GRS Code protocol is a PLT protocol, and achieves the rate $(\frac{1}{N},K-D+1)$.
\end{lemma}

\begin{proof}
Since the messages $\mathbf{X}_{[K]}$ are uniformly and independently distributed over $\mbb{F}^{S}_{q}$, and
$\{\hat{X}_1,\dots,\hat{X}_r\}$ are linearly independent combinations of the messages in $X_{[K]}$, thus $\{\hat{\mathbf{X}}_1,\dots,\hat{\mathbf{X}}_r\}$ are uniformly and independently distributed over $\mbb{F}^{S}_{q}$ as well, i.e., ${H(\hat{\mathbf{X}}_1) = \dots = H(\hat{\mathbf{X}}_r) = S\log q=B}$. Hence, the rate of the Modified GRS Code protocol is the same as the rate of the PC protocol for $N$ servers and ${K-D+1}$ messages, which is given by ${\Phi(\frac{1}{N},K-D+1)}$ (see \cite[Theorem~1]{sun2018capacity}).

From the step $4$ of the Modified GRS Code protocol, it is evident that the recoverability condition is satisfied. For the joint privacy of the proposed protocol, the proof is as follows. The PC protocol protects the privacy of the function requested by the user (i.e., no server can infer any information about the index of the function requested by the user upon receiving the query). Consider an arbitrary server $n\in [N]$, which receives an arbitrary query ${Q^{[\mc{W},V]}_n}$, generated by the proposed protocol. Given ${\mathbf{Q}^{[\bc{W},\mbf{V}]}_n = Q^{[\mc{W},V]}_n}$, from the perspective of server $n$, every function ${Y_f}$ for ${f \in [F]}$, is equally likely to be the user's desired function. We denote the support of $Y_f$ by $\mathcal{Y}_f$, i.e., $\mathcal{Y}_f$ is the set of all indices $i\in [K]$ such that $X_i$ has a non-zero coefficient in the linear combination $Y_f$. Note that for any $\mc{\tilde{W}} \in \mbb{W}$, in the proposed protocol, there exists only one function $Y_{f^{*}}$  among $Y_1,\dots,Y_F$ with $\mc{Y}_{f^{*}}=\mc{\tilde{W}}$. Thus, for any $\mc{\tilde{W}} \in \mbb{W}$ and every $n \in [N]$, the following holds 
\[
\mathbb{P}(\bc{W}= \mc{\tilde{W}}| \mbf{Q}_n^{[\bc{W},\mbf{V}]}= {Q}_{n}^{[\mc{W},{V}]}) = \Pr(\bc{W}= \mc{Y}_{f^{*}}|\mbf{Q}^{[{\bc{W}},\mbf{V}]}_n = Q_n) = \frac{1}{F}= \frac{1}{\binom{K}{D}}=\mathbb{P}(\bc{W}= \mc{\tilde{W}}).
\] This completes the proof.
\end{proof}

\bibliographystyle{IEEEtran}
\bibliography{PLT}

\end{document}